\documentclass{llncs}

\usepackage{fancyhdr}
\usepackage{amssymb,amsmath}
\usepackage[usenames,dvipsnames]{color}
\usepackage[colorlinks=true,linkcolor=blue,citecolor=Mahogany]{hyperref}

\usepackage[comma,authoryear,square,sort]{natbib} 
\usepackage{multirow}
\usepackage{mathtools}

\newcommand{\eps}{\varepsilon}
\renewcommand{\epsilon}{\eps}

\newcommand{\ignore}[1]{}

\usepackage{graphicx}

\usepackage{xspace}

\newcommand{\NPCshort}{\ensuremath{\mathsf{NPC}}\xspace}

\newcommand{\Pshort}{\ensuremath{\mathsf{P}}\xspace}
\newcommand{\NPshort}{\ensuremath{\mathsf{NP}}\xspace}

\newenvironment{proofsketch}{\noindent {\em Proof sketch:}}{ \hfill $\square$\\ }

\begin{document}

\title{Detecting Possible Manipulators in Elections}

\author{Palash Dey, Neeldhara Misra, and Y. Narahari\\ \texttt{palash@csa.iisc.ernet.in, mail@neeldhara.com, hari@csa.iisc.ernet.in}}

\institute{Department of Computer Science and Automation \\Indian Institute of Science - Bangalore, India.\\[10pt]Date: \today}

\maketitle
\pagestyle{fancy}

% \maketitle

\begin{abstract}

Manipulation is a  problem of fundamental importance 
in the context of voting in which the voters exercise their 
votes strategically instead of voting honestly to prevent
selection of an alternative that is less preferred. 
The Gibbard-Satterthwaite theorem 
shows that there is no strategy-proof voting rule that simultaneously satisfies 
certain combinations of desirable properties. 
Researchers have attempted to get around
the impossibility results in several ways such as
domain restriction and  computational hardness of manipulation. 
However these approaches have been shown to have limitations. 
Since prevention of manipulation seems to be elusive, an interesting research direction 
therefore is  detection of manipulation. Motivated by this, 
we initiate the study of detection of possible manipulators 
in an election.  

We formulate two pertinent computational 
problems - Coalitional Possible Manipulators (CPM) and 
Coalitional Possible Manipulators given Winner (CPMW), where a 
suspect group of voters is provided as input to compute whether 
they can be a potential coalition of possible manipulators. 
In the absence of any suspect group, we formulate two more computational 
problems namely Coalitional Possible Manipulators Search (CPMS), and Coalitional Possible Manipulators Search 
given Winner (CPMSW). We provide polynomial time algorithms
for these problems, for several popular voting rules. For a few other voting rules,
we show that these problems are in \NPshort{}-complete. We observe that 
detecting manipulation maybe easy even when manipulation is hard, as seen
for example, in the case of the Borda voting rule.
\end{abstract}

\keywords{Computational social choice, voting, manipulation, detection}

\newpage

\section{Introduction}

On many occasions, agents need to agree upon a common decision although 
they have different preferences over the available alternatives. A natural 
approach used in these situations is voting. Some classic examples 
of the use of voting rules in the context of multiagent systems are in 
collaborative filtering~\citep{pennock2000social}, rank aggregation in web~\citep{dwork2001rank} etc.

In a typical voting scenario, we have a set of $m$ candidates and a set of $n$ voters 
reporting their rankings of the candidates called their preferences or votes. A voting rule 
selects one candidate as the winner once all the voters provide their votes. A set of votes over a set of 
candidates along with a voting rule is called an election. 
A basic problem with voting rules is that the voters 
may vote strategically instead of voting honestly, leading to the 
selection of a candidate which is not the {\em actual winner}. 
We call a candidate actual winner if, it wins the election when every voter votes truthfully. 
This phenomenon of strategic voting is called 
{\em manipulation\/} in the context of voting. The  
Gibbard-Satterthwaite (G-S) theorem \citep{gibbard1973manipulation,satterthwaite1975strategy} 
says that manipulation is unavoidable for any \textit{unanimous}
and \textit{non-dictatorial} voting rule if we have at least three candidates. 
A voting rule is called \textit{unanimous} if 
whenever any candidate is most preferred by all the voters, such a candidate is the winner. 
A voting rule is called \textit{non-dictatorial} 
if there does not exist any voter whose most preferred candidate is always 
the winner irrespective of the votes of other voters.
The problem of manipulation is particularly relevant for multiagent systems
since agents have computational power to determine strategic votes. 
There have been several attempts to bypass the impossibility 
result of the G-S theorem.

Economists have proposed domain restriction as a way out of the impossibility 
implications of the G-S theorem. The G-S theorem assumes all possible 
preference profiles as the domain of voting rules. In a restricted domain, it has been
shown that we can have voting rules that are not vulnerable to manipulation. 
A prominent restricted domain is the domain of single peaked preferences, in 
which the median voting rule provides a satisfactory 
solution~\citep{mas1995microeconomic}. To know more about other domain restrictions, 
we refer to \citep{mas1995microeconomic,gaertner2001domain}. This approach of
restricting the domain, however, 
suffers from the requirement that the social planner needs to know the 
domain of preference profiles of the voters, which is often impractical. 

\subsection{Related Work}

Researchers in computational social choice theory have proposed invoking computational intractability of manipulation as a possible 
work around for the G-S theorem. 
Bartholdi et al.~\citep{bartholdi1991single,bartholdi1989computational} 
first proposed the idea of using computational hardness as a barrier 
against manipulation. Bartholdi et al. defined and studied the computational problem called manipulation 
where a set of manipulators have to compute their votes that make their preferred candidate win the election.
The manipulators know the votes of the truthful voters and the voting rule that will be used to compute the winner. 
Following this, a large body of research  
\citep{narodytska2011manipulation,davies2011complexity,xia2009complexity,xia2010scheduling,conitzer2007elections,obraztsova2011ties,elkind2005hybrid,faliszewski2013weighted,narodytska2013manipulating,gaspers2013coalitional,obraztsova2012optimal,faliszewski2010using,zuckerman2011algorithm,dey2014asymptotic,faliszewski2014complexity,elkind2012manipulation} 
shows that the manipulation problem is in \NPshort{}-complete (\NPCshort{}) for many voting rules. However, 
Procaccia et al.~\citep{procaccia2006junta,procaccia2007average} 
showed average case easiness of manipulation assuming \textit{junta} 
distribution over the voting profiles.  
Friedgut et al.~\citep{friedgut2008elections} showed that any \textit{neutral} voting rule which is sufficiently far from 
being dictatorial is manipulable with non-negligible probability at any uniformly random preference profile 
by a uniformly random preference. The above result holds for elections 
with three candidates only. 
A voting rule is called \textit{neutral} 
if the names of the candidates are immaterial. 
Isaksson et al.~\citep{isaksson2012geometry} 
generalize the above result to any number of candidates. 
Walsh~\citep{walsh2010empirical} empirically shows ease of manipulating an STV 
(single transferable vote) election  
-- one of the very few voting rules 
where manipulation even by one voter is in \NPCshort{}~\citep{bartholdi1991single}. 
In addition to the results mentioned above, there exist many other results in 
the literature that emphasize the weakness of considering computational complexity as a barrier 
\citep{conitzer2006nonexistence,xia2008sufficient,xia2008generalized,faliszewski2010ai,walsh2011computational}.
Hence, the barrier of computational hardness is ineffective against manipulation in many settings.

\subsection{Motivation}

In a situation where multiple attempts for prevention of manipulation fail to provide a fully satisfactory solution, detection of manipulation is a natural next step of research. There have been scenarios where establishing the occurrence of manipulation is straightforward, by observation or hindsight. For example, in sport, there have been occasions where the very structure of the rules of the game have encouraged teams to deliberately lose their matches. Observing such occurrences in, for example, football (the 1982 FIFA World Cup football match played between West Germany and Austria) and badminton (the quarter-final match between South Korea and China in the London 2012 Olympics), the relevant authorities have subsequently either changed the rules of the game (as with football) or disqualified the teams in question (as with the badminiton example). The importance of detecting manipulation lies in the potential for implementing corrective measures in the future. For reasons that will be evident soon, it is not easy to formally define the notion of manipulation detection. Assume that we have the votes from an
election that has already happened. A voter is potentially a manipulator if
there exists a preference $\succ$, different from the voter's reported preference,
which is such that the voter had an ``incentive to deviate'' from the former. 
Specifically, suppose the candidate
who wins with respect to this voter's reported preference is preferred (in $\succ$) over the 
candidate who wins with respect to $\succ$. In such a situation, $\succ$ could 
potentially be the voter's truthful preference, and the voter could be 
refraining from being truthful because an untruthful vote leads to a
more favorable outcome with respect to $\succ$. Note that we do not (and indeed, cannot) conclusively suggest that a particular voter 
has manipulated an election. This is because the said voter can always claim that she voted truthfully; 
since her actual preference is only known to her, there is no way to prove or disprove 
such a claim. Therefore, we are inevitably limited to asking only 
whether or not a voter has \textit{possibly} manipulated an election.

Despite this technical caveat, it is clear that efficient detection of manipulation, even if it is only possible manipulation, is potentially of interest in practice. We believe that, the information whether a certain group of voters have possibly manipulated an election or not would be very useful to social planners. For example, the organizers of an event, say London 2012 Olympics, maybe very interested to have this information. Also, in settings where data from many past elections (roughly over a fixed set of voters) is readily available, it is conceivable that possible manipulation could serve as suggestive evidence of real manipulation. Aggregate data about possible manipulations, although formally inconclusive, could serve as an important evidence of real manipulation\ignore{, especially in situations where the instances of possible manipulation turn out to be statistically significant. Thus, efficient detection of possible manipulation would provide a very useful input to a social planner for future elections}. We remark that having a rich history is typically not a problem, 
particularly for AI related applications, 
since the data generated from an election is normally kept for future 
requirements (for instance, for data mining or learning). For example, several past affirmatives for possible manipulation is one possible way of formalizing the notion of erratic past behavior. Also, applications where benefit of doubt maybe important, for example, elections in judiciary systems, possible manipulation detection seems useful. Thus the computational problem of detecting possible manipulation is of definite interest in this setting.

% \noindent We remark that having a rich history is typically not a problem, 
% particularly for applications based on multiagent autonomous systems, 
% since the data generated from an election is normally kept for future 
% requirements (for instance, for data mining or learning).

% Another interesting, question to ask is what is the minimum size of a possible
% manipulating coalition. If this number is large, it would indicate that a
% profile is likely to be truthful, since it would be difficult for a large
% coalition to act in a coordinated manner.

\subsection{Contributions} 

The novelty of this paper is in initiating research on {\em detection
of possible manipulators\/} in elections. We formulate four pertinent computational
problems in this context:
\begin{itemize}
 \item CPM: In the {\em coalitional possible manipulators\/} problem, 
  we are interested in whether or not a given subset of voters is a possible coalition of manipulators [Definition\nobreakspace \ref {cpmprobdef}]. 
 \item CPMW: The {\em coalitional possible manipulators given winner\/} is  
  the CPM problem with the additional information about who the winner would have been if the possible manipulators
   had all voted truthfully [Definition\nobreakspace \ref {CPMWdef}].
 \item CPMS, CPMSW: In CPMS ({\em Coalitional Possible Manipulators Search\/}), we want to know, whether there exists any coalition of possible manipulators of a size at most $k$ [Definition\nobreakspace \ref {cpmoprobdef}]. Similarly, we define CPMSW ({\em Coalitional Possible Manipulators Search given
Winner\/}) [Definition\nobreakspace \ref {cpmwoprobdef}].
\end{itemize}
Our specific findings are as follows.
\begin{itemize}
 \item We show that all the four problems above,
for scoring rules and the maximin voting rule, 
 are in \Pshort{} when the coalition size is one [Theorem\nobreakspace \ref {PUMScoringRuleP} and Theorem\nobreakspace \ref {PUMMaximinP}].
 \item We show that all the four problems, for any coalition size, are in \Pshort{} for a wide class of scoring rules 
 which include the Borda voting rule [Theorem\nobreakspace \ref {CPMWScr}, Theorem\nobreakspace \ref {thm:CPMWOScr} and Corollary\nobreakspace \ref {bordakapp}]. 
 \item We show that, for the Bucklin voting rule [Theorem\nobreakspace \ref {UPCMBucklin}], both the CPM and CPMW problems 
 are in \Pshort{}. The CPMS and CPMSW problems for the Bucklin voting rule are also in \Pshort{}, when we have 
 maximum possible coalition size $k=O(1)$.
 \item We show that both the CPM and the CPMW problems are in \NPCshort{} for the 
 STV voting rule [Theorem\nobreakspace \ref {UPCMSTVNPC} and Corollary\nobreakspace \ref {CPMWSTVNPC}]. 
 We also prove that the CPMW problem is in \NPCshort{} for maximin voting rule [Theorem\nobreakspace \ref {maxmincpmwnpc}]. 
\end{itemize}We observe that all the four problems 
are computationally easy for many voting rules that we study in this paper. This can be taken as a positive result. 
The results for the CPM and the CPMW problems are summarized in Table 1.

\begin{table}[htbp]
  \begin{center}
  {\renewcommand{\arraystretch}{1.7}
\begin{tabular}{|c|c|c|c|c|}\hline
  Voting Rule	& CPM,$c=1$	&CPM	  	&CPMW,$c=1$	&CPMW	 	\\\hline\hline
  Scoring Rules	& \Pshort{}	&?		&\Pshort{}	&?		\\\hline
  Borda		& \Pshort{}	& \Pshort{}	& \Pshort{}	&\Pshort{}	\\\hline 
  $k$-approval	& \Pshort{}	& \Pshort{}	& \Pshort{}	& \Pshort{}	\\\hline 
  Maximin	& \Pshort{}	&?		&\Pshort{}	&\NPCshort{}\\\hline
%   Copeland$^{\alpha}$	& \Pshort{}	&?	& \Pshort{}	&?		\\\hline
  Bucklin	& \Pshort{}	& \Pshort{}	& \Pshort{}	& \Pshort{}	\\\hline 
  STV		& \NPCshort{}	& \NPCshort{}	& \NPCshort{}	& \NPCshort{}	\\\hline
 \end{tabular}}
 \caption{\normalfont Results for CPM and CPMW ($c$ denotes coalition size). The `?' mark means that the problem is open.}
 \end{center} 
\end{table} 

This paper is a significant extension of the conference version of this work~\cite{detection}: this extended version includes all the proofs.%, and the results on voting trees, plurality with runoff, and k-approval are new. (The conference version also did not mention plurality and veto; these results are easy and follow from known results, as explained in the footnotes under the table.)

\paragraph{Organization}

The rest of the paper is organized as follows. We describe the necessary preliminaries in Section~\ref{sec:prelim}; we formally define the 
computational problems in Section~\ref{sec:prob_def}; we present the results in Section~\ref{sec:results} and finally we conclude in Section~\ref{sec:con}.

\section{Preliminaries}\label{sec:prelim}

Let $\mathcal{V}=\{v_1, \dots, v_n\}$ be the set of all \emph{voters} and $\mathcal{C}=\{c_1, \dots, c_m\}$ 
the set of all \emph{candidates}. 
Each voter $v_i$'s \textit{vote} is a \emph{preference} $\succ_i$ over the 
candidates which is a linear order over $\mathcal{C}$. 
For example, for two candidates $a$ and $b$, $a \succ_i b$ means that the voter $v_i$ prefers $a$ to $b$. 
We will use $a >_i b$ to denote the fact that $ a \succ_i b, a \ne b$. 
We denote the set of all linear orders over $\mathcal{C}$ by $\mathcal{L(C)}$. 
Hence, $\mathcal{L(C)}^n$ denotes the set of all $n$-voters' preference profile $(\succ_1, \cdots, \succ_n)$. 
We denote the $(n-1)$-voters' preference profile by $(\succ_1, \cdots, \succ_{i-1}, \succ_{i+1}, \cdots, \succ_n)$ 
by $\succ_{-i}$. 
We denote the set $\{1, 2, 3, \dots\}$ by 
$\mathbb{N}^+$. The power set of $\mathcal{C}$ is denoted by $2^\mathcal{C}$, and $\emptyset$ 
denotes the empty set. 
A map $r_c:\cup_{n,|\mathcal{C}|\in\mathbb{N}^+}\mathcal{L(C)}^n\longrightarrow 2^\mathcal{C}\setminus\{\emptyset\}$
is called a \emph{voting correspondence}. A map 
$t:\cup_{|\mathcal{C}|\in\mathbb{N}^+}2^\mathcal{C}\setminus\{\emptyset\}\longrightarrow \mathcal{C}$ is called a \emph{tie breaking rule}.
Commonly used tie breaking rules are \emph{lexicographic} tie breaking rules where ties are broken 
according to a predetermined preference $\succ_t \in \mathcal{L(C)}$. 
A \emph{voting rule} is $r=t\circ r_c$, where $\circ$ denotes composition of mappings. 
Given an election $E$, we can construct a weighted graph $G_E$ called 
\textit{weighted majority graph} from $E$. The set of vertices in $G_E$ is the set of candidates in $E$. 
For any two candidates $x$ and $y$, the weight on the edge $(x,y)$ is $D_E(x,y) = N_E(x,y) - N_E(y,x)$, 
where $N_E(x,y)(\text{respectively }N_E(y,x))$ is the number of voters who prefer $x$ to $y$ (respectively $y$ to $x$). 
Some examples of common voting correspondences are as follows.
\begin{itemize}
 \item {\bf Positional scoring rules:} Given an $m$-dimensional vector $\vec{\alpha}=\left(\alpha_1,\alpha_2,\dots,\alpha_m\right)\in\mathbb{R}^m$ 
 with $\alpha_1\ge\alpha_2\ge\dots\ge\alpha_m$ and $\alpha_1 > \alpha_m$, we can naturally define a
 voting rule - a candidate gets score $\alpha_i$ from a vote if it is placed at the $i^{th}$ position, and the 
 score of a candidate is the sum of the scores it receives from all the votes. 
 The winners are the candidates with maximum score. A scoring rule is called a 
 strict scoring rule if $\alpha_1>\alpha_2>\dots>\alpha_m$.
 For $\alpha=\left(m-1,m-2,\dots,1,0\right)$, we get the \emph{Borda} voting rule. With $\alpha_i=1$ $\forall i\le k$ and $0$ else, 
 the voting rule we get is known as the $k$-\emph{approval} voting rule. The \emph{plurality} voting rule is the $1$-\emph{approval} voting rule and the \emph{veto} voting rule is the $(m-1)$-\emph{approval} voting rule. 
 \item {\bf Maximin:} The maximin score of a candidate $x$ is $min_{y\ne x} D(x,y)$. The winners are the candidates with maximum maximin score.
 \item {\bf Bucklin:} A candidate $x$'s Bucklin score is the minimum number $l$ such that at least half 
 of the voters rank $x$ in their top $l$ positions. The winners are the candidates with lowest Bucklin score. This voting rule is 
 also sometimes referred as the simplified Bucklin voting rule.
 \item {\bf Single Transferable Vote:} In Single Transferable Vote (STV), 
 a candidate with least plurality score is dropped out of the election and its votes 
 are transferred to the next preferred candidate. If two or more candidates receive least plurality score, then some predetermined tie breaking rule is used. The candidate that remains after $(m-1)$ rounds is the winner.
\end{itemize}

\section{Problem Formulation}\label{sec:prob_def}

Consider an election that has already happened in which all the votes are known and thus the 
winner $x\in\mathcal{C}$ is also known. We call the candidate $x$ the {\em current winner} of the election. 
The authority may suspect that the voters belonging to $M\subset \mathcal{V}$ have formed a coalition among themselves and manipulated the 
election by voting non-truthfully. The authority believes that other voters who do not belong to $M$, have voted truthfully. 
We denote the coalition size $|M|$ by $c$. 
Suppose the authority has auxiliary information, 
maybe from some other sources, which says that the \textit{actual winner} should have been some candidate 
$y\in \mathcal{C}$ other than $x$. 
We call a candidate {\em actual winner} if it wins the election where all the voters vote truthfully. This means that the authority
thinks that, had the voters in $M$ voted truthfully, the candidate $y$ would have been the winner.  
We remark that there are practical situations, for example, 1982 FIFA World cup or 2012 London Olympics, where 
the authority knows the actual winner. 
This situation is formalized below.

\begin{definition}\label{cpmwdef}
Let $r$ be a voting rule, and $(\succ_i)_{i\in \mathcal{V}}$ be a voting profile of a set $\cal V$ of $n$ voters. Let $x$ be the winning candidate with respect to $r$ for this profile. For a candidate $y \neq x$, $M\subset \mathcal{V}$ is said to be a coalition of possible manipulators against $y$ with respect to $r$ if there exists a $|M|$-voters' profile 
 $(\succ_j^{\prime})_{j\in M} \in \mathcal{L(C)}^{|M|}$ such that $x \succ_j^\prime y, \forall j\in M$, and further,
$r((\succ_j)_{j\in \mathcal{V}\setminus M},(\succ_i^{\prime})_{i\in M}) = y$.  
\end{definition} 

Using the notion of coalition of possible manipulators, we formulate a computational problem called {\em Coalitional Possible Manipulators given Winner (CPMW)} as follows.

\begin{definition}{\bf(CPMW Problem)}\label{CPMWdef}\\
Given a voting rule $r$, a preference profile $(\succ_i)_{i\in \mathcal{V}}$ of a set of voters $\mathcal{V}$ over a set of candidates $\mathcal{C}$, a subset of voters $M\subset \mathcal{V}$, and a candidate $y$, determine if $M$ is a coalition of possible manipulators against $y$ with respect to $r$.
\end{definition}

In the CPMW problem, the actual winner is given in the input. However, 
it may very well happen that the authority does not have any other information to 
guess the \textit{actual winner} - the candidate who would have won the election had the voters in $M$ 
voted truthfully. In this situation, the authority is interested in knowing whether there 
is a $|M|$-voter profile which along with the votes in $\mathcal{V}\setminus M$ makes some candidate 
$y\in \mathcal{C}$ the winner who is different from the current winner $x\in \mathcal{C}$ and all 
the preferences in the $|M|$-voters' profile prefer $x$ to $y$. If such a $|M|$-voter profile exists for the subset of voters $M$, then we call $M$ a \textit{coalition of possible manipulators} and the 
corresponding computational problem is called \textit{Coalitional Possible Manipulators (CPM)}. 
These notions are formalized below.

\begin{definition}%{{\bf Coalition of Possible Manipulators (with respect to $r$)}}
\label{cpmdef}
Let $r$ be a voting rule, and $(\succ_i)_{i\in \mathcal{V}}$ be a voting profile of a set $\cal V$ of $n$ voters. A subset of voters $M\subset \mathcal{V}$ is called a coalition of possible manipulators with respect to $r$ if $M$ is a coalition of possible manipulators against some candidate $y$ with respect to $r$.
\end{definition}

\begin{definition}{\bf(CPM Problem)}\label{cpmprobdef}\\
Given a voting rule $r$, a preference profile $(\succ_i)_{i\in \mathcal{V}}$ of a set of voters $\mathcal{V}$ over a set of candidates $\mathcal{C}$, and a subset of voters $M\subset \mathcal{V}$, determine if $M$ is a coalition of possible manipulators with respect to $r$.
\end{definition}

In both the CPMW and CPM problems, a subset of voters which the authority suspect to be a coalition of manipulators, is given in the input. However, there can be situations where there is no specific subset of voters to suspect. In those scenarios, it may still be useful to know, what are the possible coalition of manipulators of size less than some number $k$. Towards that end, we extend the CPMW and CPM problems to search for a coalition of potential possible manipulators and call them {\em Coalitional  Possible Manipulators Search given Winner (CPMSW)} and {\em Coalitional Possible Manipulators Search (CPMS)} respective.

\begin{definition}{\bf(CPMSW Problem)}\label{cpmwoprobdef}\\
Given a voting rule $r$, a preference profile $(\succ_i)_{i\in \mathcal{V}}$ of a set of voters $\mathcal{V}$ over a set of candidates $\mathcal{C}$, a candidate $y$, and an integer $k$, determine whether there exists any $M\subset \mathcal{V}$ with $|M|\le k$ such that $M$ is a coalition of possible manipulators against $y$.
\end{definition}

\begin{definition}{\bf(CPMS Problem)}\label{cpmoprobdef}\\
Given a voting rule $r$, a preference profile $(\succ_i)_{i\in \mathcal{V}}$ of a set of voters $\mathcal{V}$ over a set of candidates $\mathcal{C}$, and an integer $k$, determine whether there exists any $M\subset \mathcal{V}$ with $|M|\le k$ such that $M$ is a coalition of possible manipulators.
\end{definition}

\subsection{Discussion} 

The CPMW problem may look very similar to the manipulation problem~\citep{bartholdi1989computational,conitzer2007elections}- in both the problems a set of voters try to make a candidate winner. However, in the CPMW problem, the actual winner must be less preferred to the current winner. Although it may look like a subtle difference, it changes the nature and complexity theoretic behavior of the problem completely. For example, we show that all the four problems have an efficient algorithm for a large class of voting rules that includes the Borda voting rule for which the manipulation problem is in \NPCshort{}, even when we have at least two manipulators~\citep{davies2011complexity,betzler2011unweighted}. Another important difference is that the manipulation problem, in contrast to the problems studied in this paper, does not take care of manipulators' preferences. We believe that there does not exist any formal reduction between the CPMW problem and the manipulation problem. 

On the other hand, the CPMS problem is similar to the margin of victory problem defined by Xia~\citep{xia2012computing}, where also we are looking for changing the current winner by changing at most some $k$ number of votes, which in turn identical to the destructive bribery problem~\citep{faliszewski2009hard}. Whereas, in the CPMS problem, the vote changes can occur in a restricted fashion. Here also, the margin of victory problem has the hereditary property which the CPMS problem does not possess. These two problems do not seem to have any obvious complexity theoretic implications.

Now, we explore the connections among the four problems that we study here. Notice that, a polynomial time algorithm for the CPM and the CPMW problems gives us a polynomial time algorithm for the CPMS and the CPMSW problems for any maximum possible coalition size $k=O(1)$. Also, observe that, a polynomial time algorithm for the CPMW (respectively CPMSW) problem implies a polynomial time algorithm for the CPM (respectively CPMS) problem. Hence, we have the following propositions.

\begin{proposition}\label{cpmcpmw}
 For every voting rule, if the maximum possible coalition size $k=O(1)$, then,
 \[CPMW\in\text{\Pshort{}} \Rightarrow CPM, CPMSW, CPMS\in\text{\Pshort{}}\]
\end{proposition}

\begin{proposition}\label{prop:cpmocpmwo}
 For every voting rule,
 \[CPMSW\in\text{\Pshort{}} \Rightarrow CPMS\in\text{\Pshort{}}\]
\end{proposition}

\section{Results}\label{sec:results}

In this section, we present our algorithmic results for the CPMW, CPM, CPMSW, and CPMS problems for various 
voting rules.

\subsection{Scoring Rules}

Below we have certain lemmas which form a crucial ingredient of our algorithms. To begin with, we define the notion of a {\em manipulated preference}. Let $r$ be a scoring rule and $\succ:=(\succ_i,\succ_{-i})$ be a voting profile of $n$ voters. Let $\succ_i^{\prime}$ be a preference such that $$r(\succ) >_i^{\prime}  r(\succ_i^{\prime},\succ_{-i}).$$ Then, we say that $\succ_i^{\prime}$ is a $(\succ,i)$-manipulated preference with respect to $r$. We omit the reference to $r$ if it is clear from the context.

\begin{lemma}\label{losersorted}
Let $r$ be a scoring rule and $\succ:=(\succ_i,\succ_{-i})$ be a voting profile of $n$ voters.  Let $a$ and $b$ be two candidates such that $score_{\succ_{-i}}(a) > score_{\succ_{-i}}(b)$, and let $\succ_i^{\prime}$ be $(\succ,i)$-manipulated preference where $a$ precedes $b$:
$$\succ_i^{\prime}:= \cdots > a > \cdots > b > \cdots$$
If $a,b$ are not winners with respect to either $(\succ_i^{\prime},\succ_{-i})$ or  $\succ$, then the preference $\succ_i^{\prime\prime}$ obtained from $\succ_i^{\prime}$ by interchanging $a$ and $b$ is also $(\succ,i)$-manipulated. 
\end{lemma}

\begin{proof}
Let $x:= r(\succ_i^{\prime},\succ_{-i})$. If suffices to show that $x$ continues to win in the proposed profile $(\succ_i^{\prime\prime},\succ_{-i})$. To this end, it is enough to argue the scores of $a$ and $b$ with respect to $x$. First, consider the score of $b$ in the new profile:
 \begin{eqnarray}
  score_{(\succ_i^{\prime\prime},\succ_{-i})}(b) 
  &=& score_{\succ_i^{\prime\prime}}(b) + score_{\succ_{-i}}(b) \nonumber \\
  &<& score_{\succ_i^{\prime}}(a) + score_{\succ_{-i}}(a) \nonumber \\
  &=& score_{(\succ_i^{\prime},\succ_{-i})}(a) \nonumber \\
  &\le& score_{(\succ_i^{\prime},\succ_{-i})}(x) \nonumber \\
  &=& score_{(\succ_i^{\prime\prime},\succ_{-i})}(x) \label{c'notwin}
 \end{eqnarray}
 The second line uses the fact that 
 $score_{\succ_i^{\prime\prime}}(b) = score_{\succ_i^{\prime}}(a)$ and 
 $score_{\succ_{-i}}(b) < score_{\succ_{-i}}(a)$. The fourth line comes from 
 the fact that $x$ is the winner and the last line follows from the fact that 
 the position of $x$ is same in both profiles. Similarly, we have the following argument for the score of $a$ in the new profile (the second line below simply follows 
 from the definition of scoring rules).
 \begin{eqnarray}
  score_{(\succ_i^{\prime\prime},\succ_{-i})}(a) 
  &=& score_{\succ_i^{\prime\prime}}(a) + score_{\succ_{-i}}(a) \nonumber \\
  &\le& score_{\succ_i^{\prime}}(a) + score_{\succ_{-i}}(a) \nonumber \\
  &=& score_{(\succ_i^{\prime},\succ_{-i})}(a) \nonumber \\
  &\le& score_{(\succ_i^{\prime},\succ_{-i})}(x) \nonumber \\
  &=& score_{(\succ_i^{\prime\prime},\succ_{-i})}(x) \label{cnotwin}
 \end{eqnarray}
 Since the tie breaking rule is according to some predefined fixed order 
 $\succ_t  \in \mathcal{L(C)}$ and the candidates tied with winner in 
 $(\succ_i^{\prime\prime},\succ_{-i})$ also tied with winner in 
 $(\succ_i^{\prime},\succ_{-i})$, we have the following,
 \[ r(\succ) >_i^{\prime\prime} r(\succ_i^{\prime\prime},\succ_{-i})\]
\qed\end{proof}
\setcounter{equation}{0}

We now show that, if there is some $(\succ,i)$-manipulated preference with respect to a scoring rule $r$, then there exists a $(\succ,i)$-manipulated preference with a specific structure.

\begin{lemma}\label{algotheorem}
 Let $r$ be a scoring rule and $\succ:=(\succ_i,\succ_{-i})$ be a voting profile of $n$ voters. If there is some $(\succ,i)$-manipulated preference with respect to $r$, then there also exists a $(\succ,i)$-manipulated preference $\succ_i^{\prime}$ where the actual winner $y$ immediately follows the current winner $x$:
  $$\succ_i^{\prime}:= \cdots > x > y > \cdots $$
 and the remaining candidates are in nondecreasing ordered of their scores from ${\succ_{-i}}$.
\end{lemma}

\begin{proof}
Let $\succ^{\prime\prime}$ be a $(\succ,i)$-manipulated preference with respect to $r$.  
 Let $x:= r(\succ), y:= r(\succ^{\prime\prime},\succ_{-i})$. From Lemma\nobreakspace \ref {losersorted}, 
 without loss of generality, we may assume that, all candidates except $x, y$ 
 are in nondecreasing order of $score_{\succ_{-i}}(.)$ in the preference $\succ^{\prime\prime}$. 
 If $\succ_i^{\prime\prime}:= \cdots \succ x \succ \cdots \succ y \succ \cdots \succ \cdots$, 
 we define $\succ_i^{\prime}:= \cdots \succ x \succ y \succ \cdots \succ \cdots$ from 
 $\succ_i^{\prime\prime}$ where 
 $y$ is moved to the position following $x$ and the position of the candidates in between 
 $x$ and $y$ in $\succ_i^{\prime\prime}$ is deteriorated by one position each. The position 
 of the rest of the candidates remain same in both $\succ_i^{\prime\prime}$ and $\succ_i^{\prime}$.
 Now we have following,
 \begin{eqnarray}
  score_{(\succ_i^{\prime},\succ_{-i})}(y) 
  &=& score_{\succ_i^{\prime}}(y) + score_{\succ_{-i}}(y) \nonumber \\
  &\ge& score_{\succ_i^{\prime\prime}}(y) + score_{\succ_{-i}}(y) \nonumber \\
  &=& score_{(\succ_i^{\prime\prime},\succ_{-i})}(y)
 \end{eqnarray}
 We also have,
 \begin{eqnarray}
  score_{(\succ_i^{\prime},\succ_{-i})}(a) 
  \le score_{(\succ_i^{\prime\prime},\succ_{-i})}(a), 
  \forall a \in \mathcal{C} \setminus \{y\}
 \end{eqnarray}
 Since the tie breaking rule is according to some predefined order $\succ_t  \in \mathcal{L(C)}$, 
 we have the following,
 \[ r(\succ) >_i^{\prime} r(\succ^{\prime},\succ_{-i}) \]
\qed\end{proof}

Using Lemmas\nobreakspace \ref {losersorted} and\nobreakspace  \ref {algotheorem}, we now present the results for the scoring rules.

\begin{theorem}\label{PUMScoringRuleP}
 When $c$=1, the CPMW, CPM, CPMSW, and CPMS problems for scoring rules are in \Pshort{}.
\end{theorem}

\begin{proof}
From Proposition\nobreakspace \ref {cpmcpmw}, it is enough to give a polynomial time algorithm for the CPMW problem. So consider the CPMW problem. We are given the actual winner $y$ and we compute the current winner $x$ with respect to $r$. Let $\succ_{[j]}$ be a preference where $x$ and $y$ are in positions $j$ and $(j+1)$ respectively, and the rest of the candidates are in nondecreasing order of the score that they receive from $\succ_{-i}$. For $j \in \{1,2,\ldots,m-1\}$, we check if $y$ wins with the profile $(\succ_{-i},\succ_{[j]})$. If we are successful with at least one $j$ we report YES, otherwise we say NO. The correctness follows from Lemma\nobreakspace \ref {algotheorem}. 
Thus we have a polynomial time algorithm for CPMW when $c=1$ and the theorem follows from Proposition\nobreakspace \ref {cpmcpmw}.
\qed\end{proof}

Now, we study the CPMW and the CPM problems when $c>1$. 
If $m=O(1)$, then both the CPMW and the CPM problems for any anonymous and efficient voting rule $r$ can be solved in polynomial time by iterating over all possible ${m!+c-1 \choose m!}$ ways the manipulators can have actual preferences. A voting rule is called efficient if winner determination under it is in \Pshort{}.

\begin{theorem}\label{CPMWScr}
 For scoring rules with $\alpha_1 - \alpha_2 \le \alpha_i - \alpha_{i+1}, \forall i$, 
 the CPMW and the CPM problems are in \Pshort{}.
\end{theorem}

\begin{proof}
 We provide a polynomial time algorithm for the CPMW problem in this setting. 
 Let $x$ be the current winner and $y$ be the given actual winner. Let
 $M$ be the given subset of voters. Let $((\succ_i)_{i\in M}, (\succ_j)_{j\in V\setminus M})$ be the 
 reported preference profile. Without loss of generality, we assume that $x$ is the most preferred candidate 
 in every $\succ_i, i\in M$. Let us define $\succ_i^{\prime}, i\in M,$ by moving $y$ to the second position in the preference
 $\succ_i$. In the profile $((\succ_i^{\prime})_{i\in M}, (\succ_j)_{j\in V\setminus M})$, the winner is either $x$ or $y$ since 
 only $y$'s score has increased. We claim that $M$ is a coalition of possible manipulators with respect to $y$ 
 \textit{if and only if} $y$ is the winner in preference profile 
 $((\succ_i^{\prime})_{i\in M}, (\succ_j)_{j\in V\setminus M})$. This can be seen as follows. Suppose there exist 
 preferences $\succ_i^{\prime\prime}, $ with $x \succ_i^{\prime\prime} y, i\in M,$ for which $y$ wins in the profile 
 $((\succ_i^{\prime\prime})_{i\in M}, (\succ_j)_{j\in V\setminus M})$. Now without loss of generality, we can assume that 
 $y$ immediately follows $x$ in all $\succ_i^{\prime\prime}, i\in M,$ and  
 $\alpha_1 - \alpha_2 \le \alpha_i - \alpha_{i+1}, \forall i$ implies that we can also assume that $x$ and $y$ are in the 
 first and second positions respectively in all $\succ_i^{\prime\prime}, i\in M$. Now in both the profiles, 
 $((\succ_i^{\prime})_{i\in M}, (\succ_j)_{j\in V\setminus M})$ and $((\succ_i^{\prime\prime})_{i\in M}, (\succ_j)_{j\in V\setminus M})$, the score of $x$ and $y$ are same. But in the first profile $x$ wins and in the second profile $y$ wins, which is a contradiction.
\qed\end{proof}

\begin{theorem}\label{thm:CPMWOScr}
 For scoring rules with $\alpha_1 - \alpha_2 \le \alpha_i - \alpha_{i+1}, \forall i$, 
 the CPMSW and the CPMS problems are in \Pshort{}.
\end{theorem}

\begin{proof}
 From Proposition\nobreakspace \ref {prop:cpmocpmwo}, it is enough to prove that $CPMSW \in \mathcal{P}$. Let $x$ be the current winner, $y$ be the given actual winner and $s(x)$ and $s(y)$ be their current respective scores. For each vote $v\in \mathcal{V}$, we compute a number $\Delta(v) = \alpha_2 - \alpha_j - \alpha_1 + \alpha_i$, where $x$ and $y$ are receiving scores $\alpha_i$ and $\alpha_j$ respectively from the vote $v$. Now, we output yes iff there are $k$ votes $v_i, 1\le i\le k$ such that, $\sum_{i=1}^k \Delta(v_i) \ge s(x) - s(y)$, which can be checked easily by sorting the $\Delta(v)$'s in nonincreasing order and checking the condition for the first $k$ $\Delta(v)$'s, where $k$ is the maximum possible coalition size specified in the input. The proof of correctness follows by exactly in the same line of argument as the proof of Theorem\nobreakspace \ref {CPMWScr}.
\qed\end{proof}

For the plurality voting rule, we can solve all the problems easily using max flow. Hence, from Theorem\nobreakspace \ref {CPMWScr} and Theorem\nobreakspace \ref {thm:CPMWOScr}, we have the following result.

\begin{corollary}\label{bordakapp}
 The CPMW, CPM, CPMSW, and CPMS problems for the Borda and $k$-approval voting rules are in \Pshort{}.
\end{corollary}

\subsection{Maximin Voting Rule}

For the maximin voting rule, we show that all the four problems are in \Pshort{}, when we have a coalition size of one to check for.

\begin{theorem}\label{PUMMaximinP}
 The CPMW, CPM, CPMSW, and CPMS problems for maximin voting rule are in \Pshort{} for coalition size $c=1$ or maximum possible coalition size $k=1$.
\end{theorem}

\begin{proof}
 Given a $n$-voters' profile $\succ \in \mathcal{L(C)}^n$ and a voter $v_i$, 
 let the \textit{current winner} be $x:= r(\succ)$ and the given \textit{actual winner} be $y$. We will construct 
 $\succ^{\prime} = (\succ_i^{\prime} , \succ_{-i})$, if it exists, such that 
 $r(\succ) >_i^{\prime} r(\succ^{\prime})=y$, thus deciding whether $v_i$ is a 
 possible manipulator or not. Now, the maximin score of $x$ and $y$ in the profile $\succ^{\prime}$ can take one of 
 values from the set $\{ score_{\succ_{-i}} (x) \pm 1\}$ and $\{ score_{\succ_{-i}} (y) \pm 1\}$. 
 The algorithm is as follows. We first guess the maximin score of $x$ and $y$ in the profile $\succ^{\prime}$. 
 There are only four possible guesses. Suppose, we guessed that $x$'s score will decrease by one and $y$'s score 
 will decrease by one assuming that this guess makes $y$ win. 
 Now notice that, without loss of generality, 
 we can assume that $y$ immediately follows $x$ in the preference $\succ_i^{\prime}$ since $y$ is the winner 
 in the profile $\succ^{\prime}$. This implies that there are only $O(m)$ many possible positions for $x$ and $y$ 
 in $\succ_i^{\prime}$. We guess the position of $x$ and thus the position of $y$ in $\succ_i^{\prime}$.  Let $B(x)$ 
 and $B(y)$ be the sets of candidates with whom $x$ and respectively $y$ performs worst. Now since, $x$'s score 
 will decrease and $y$'s score will decrease, we have the following constraint on $\succ_i^{\prime}$. There must be a candidate each from $B(y)$ and $B(x)$ that will precede $x$. We do not know a-priori if there is one candidate that will serve as a witness for both $B(x)$ and $B(y)$, or if there separate witnesses. In the latter situation, we also do not know what order they appear in. Therefore we guess if there is a common candidate, and if not, we guess the relative ordering of the distinct candidates from $B(x)$ and $B(y)$.
  Now we place any candidate at the top position of $\succ_i^{\prime}$ if this action does not make $y$ lose the election. If there are many choices, we prioritize in favor of candidates from $B(x)$ and $B(y)$ --- in particular, we focus on the candidates common to $B(x)$ and $B(y)$ if we expect to have a common witness, otherwise, we favor a candidate from one of the sets according to the guess we start with. If still there are multiple choices, we 
 pick arbitrarily.  After that we move on to the next position, and do the same thing (except we stop prioritizing explicitly for $B(x)$ and $B(y)$ once we have at least one witness from each set). The other situations can be handled similarly with minor modifications. In this way, if it is able to get a 
 complete preference, then it checks whether $v_i$ is a possible manipulator or not 
 using this preference. If \textit{yes}, then it returns YES. Otherwise, it tries 
 other positions for $x$ and $y$ and other possible scores of $x$ and $y$. After 
 trying all possible guesses, if it cannot find the desired preference, 
 then it outputs NO. Since there are only polynomial 
 many possible guesses, this algorithm runs in a polynomial amount of time. 
 The proof of correctness follows from the proof of Theorem 1 in~\citep{bartholdi1989computational}.
\qed\end{proof}

We now show that the CPMW problem for maximin voting rule is in \NPCshort{} when 
we have $c>1$. Towards that, we use the fact that the \textit{unweighted 
coalitional manipulation (UCM)} problem for maximin voting rule is in \NPCshort{} 
\citep{xia2009complexity}, when we have $c>1$. The UCM problem is as follows. 

\begin{definition}{(UCM Problem)}\\
 Given a voting rule $r$, a set of manipulators $M\subset \mathcal{V}$, a profile of 
 non-manipulators' vote $(\succ_i)_{i\in \mathcal{V}\setminus M}$, 
 and a candidate $z\in \mathcal{C}$, we are asked whether there exists a 
 profile of manipulators' votes $(\succ_j^{\prime})_{j\in M}$ such that 
 $r((\succ_i)_{i\in \mathcal{V}\setminus M}, (\succ_j^{\prime})_{j\in M}) = z$. 
 Assume that ties are broken in favor of $z$.
\end{definition}

We define a restricted version of the UCM problem called R-UCM as follows. 

\begin{definition}{(R-UCM Problem)}\\
 This problem is the same as the UCM problem with a given guarantee - let $c:= |M|$. 
 The candidate $z$ loses pairwise election with every other candidate by $4c$ 
 votes. For any two candidates $a,b \in \mathcal{C}$, either $a$ and $b$ ties or 
 one wins pairwise election against the other one by margin of either $2c+2$ or 
 of $4c$ or of $8c$. We denote the margin by which a candidate $a$ defeats $b$, by $d(a,b)$.
\end{definition}

The R-UCM problem for maximin voting rule is in \NPCshort{} \citep{xia2009complexity}, 
when we have $c>1$. We will need the following lemma to manipulate the pairwise difference scores in the reduction. 
The lemma has been used before \citep{mcgarvey1953theorem,xia2008determining}.

\begin{lemma}\label{maxminexist}
 For any function $f:\mathcal{C} \times \mathcal{C} \longrightarrow \mathbb{Z}$, such that
 \begin{enumerate}
  \item $\forall a,b \in \mathcal{C}, f(a,b) = -f(b,a)$.
  \item $\forall a,b \in \mathcal{C}, f(a,b)$ is even,
 \end{enumerate}
 there exists a $n$ voters' profile such that for all $a,b \in \mathcal{C}$, $a$ defeats 
 $b$ with a margin of $f(a,b)$. Moreover, 
 $$n = O\left(\sum_{\{a,b\}\in \mathcal{C}\times\mathcal{C}} |f(a,b)|\right)$$
\end{lemma}

\begin{theorem}\label{maxmincpmwnpc}
 The CPMW problem for maximin voting rule is in \NPCshort{}, for $c>1$. 
\end{theorem}

\begin{proof}
Clearly the CPMW problem for maximin voting rule is in \NPshort{}. We provide a many-one reduction from the 
 R-UCM problem for the maximin voting rule to it. Given a R-UCM problem 
 instance, we define a CPMW problem instance 
 $\Gamma = (\mathcal{C}^{\prime},(\succ_i^{\prime})_{i\in\mathcal{V^{\prime}}},M^{\prime})$ as follows.
 \[ \mathcal{C}^{\prime}:= \mathcal{C} \cup \{w,d_1,d_2,d_3\}\]
 We define $\mathcal{V^{\prime}}$ such that $d(a,b)$ is the same as the R-UCM instance, 
 for all $a,b\in \mathcal{C}$ and $d(d_1,w)=2c+2, d(d_1,d_2)=8c, d(d_2,d_3)=8c, d(d_3,d_1)=8c$. 
 The existence of such a $\mathcal{V^{\prime}}$ is guaranteed from Lemma\nobreakspace \ref {maxminexist}. Moreover, 
 Lemma\nobreakspace \ref {maxminexist} also ensures that $|\mathcal{V^{\prime}}|$ is $O(mc)$. 
 The votes of the voters in $M$ is $w\succ \dots$. Thus the \textit{current winner} is $w$. 
 The \textit{actual winner} is defined to be $z$. The tie breaking rule is $\succ_t = w\succ z\succ \dots$, 
 where $z$ is the candidate whom the manipulators in $M$ want to make winner in the 
 R-UCM problem instance. Clearly this reduction takes polynomial amount of time. 
 Now we show that, $M$ is a coalition of possible 
 manipulators \textit{iff} $z$ can be made a winner.
 
 The \textit{if} part is as follows. Let $\succ_i, i\in M$ be the votes that make $z$ win. 
 We can assume that $z$ is the most preferred candidate in all the preferences $\succ_i, i\in M$.
 Now consider the preferences for the voters in $M$ is follows.
 \[\succ_i^{\prime}:= d_1\succ d_2\succ d_3\succ w\succ_i, i\in M\]
 The score of every candidate in $\mathcal{C}$ is not more than $z$. The score 
 of $z$ is $-3c$. The score of $w$ is $-3c-2$ and the scores of $d_1, d_2,$ and $d_3$ 
 are less than $-3c$. Hence, $M$ is a coalition of possible manipulators with the actual preferences 
 $\succ_i^{\prime}:= d_1\succ d_2\succ d_3\succ w\succ_i, i\in M$.
 
 The \textit{only if} part is as follows. Suppose $M$ is a coalition of possible manipulators 
 with actual preferences $\succ_i^{\prime}, i\in M$. Consider the preferences $\succ_i^{\prime}, i\in M$, 
 but restricted to the set $\mathcal{C}$ only. Call them $\succ_i, i\in M$. We claim that 
 $\succ_i, i\in M$ with the votes from $\mathcal{V}$ makes $z$ win the election. 
 If not then, there exists a candidate, say $a\in \mathcal{C}$, whose score is strictly more than 
 the score of $z$ - this is so because the tie breaking rule is in favor of $z$. 
 But this contradicts the fact that $z$ wins the election when the voters 
 in $M$ vote $\succ_i^{\prime}, i\in M$ along with the votes from $\mathcal{V^{\prime}}$.
\qed\end{proof}

\subsection{Bucklin Voting Rule}

In this subsection, we design polynomial time algorithms for both the CPMW and the CPM 
problem for the Bucklin voting rule. Again, we begin by showing that if there are profiles witnessing manipulation, then there exist profiles that do so with some additional structure, which will subsequently be exploited by our algorithm.

\begin{lemma}\label{lemmaBucklin}
Consider a preference profile $(\succ_i)_{i \in \cal{V}}$, where $x$ is the winner with respect to the Bucklin voting rule. Suppose a subset of voters $M \subset \mathcal{V}$ form a coalition of possible manipulators. Let $y$ be the actual winner. Then there exist preferences $(\succ_i^\prime)_{i \in M}$ such that $y$ is a Bucklin winner in $((\succ_i)_{i \in \cal{V} \setminus M},(\succ_i^\prime)_{i \in M})$, and further:
 \begin{enumerate}
  \item $y$ immediately follows $x$ in each $\succ_i^{\prime}$.
  \item The rank of $x$ in each $\succ_i^{\prime}$ is in one of the 
  following - first, $b(y)-1$, $b(y)$, $b(y)+1$, where $b(y)$ be the Bucklin score of $y$ in $((\succ_i)_{i \in \cal{V} \setminus M},(\succ_i^\prime)_{i \in M})$.
 \end{enumerate}
\end{lemma}

\begin{proof}
% The \textit{if} part directly follows from Definition\nobreakspace \ref {cpmdef}. The \textit{only if} part is as follows. 
 From Definition\nobreakspace \ref {cpmdef}, $y$'s rank must be worse than $x$'s rank in each $\succ_i^{\prime}$. 
 We now exchange the position of $y$ with the candidate which immediately follows $x$ in 
 $\succ_i^{\prime}$. This process does not decrease Bucklin score of any candidate except 
 possibly $y$'s, and $x$'s score does not increase. Hence $y$ will continue to win and thus  
 $\succ_i^{\prime}$ satisfies the first condition.
 
 Now to begin with, we assume that $\succ_i^{\prime}$ satisfies the first condition. 
 If the position of $x$ in $\succ_i^{\prime}$ is $b(y)-1$ or $b(y)$, we do not change it. 
 If $x$ is above $b(y)-1$ in $\succ_i^{\prime}$, then move $x$ and $y$ at the first and 
 second positions respectively. Similarly if $x$ is below $b(y)+1$ in $\succ_i^{\prime}$, 
 then move $x$ and $y$ at the $b(y)+1$ and $b(y)+2$ positions respectively. This process 
 does not decrease score of any candidate except $y$ because the Bucklin score of $x$ is at least 
 $b(y)$. The transformation cannot increase the score $y$ since its position has only been 
 improved. Hence $y$ continues to win and thus $\succ_i^{\prime}$ satisfies the second condition. 
\qed\end{proof}

Lemma\nobreakspace \ref {lemmaBucklin} leads us to the following theorem.

\begin{theorem}\label{UPCMBucklin}
 The CPMW problem and the CPM problems for Bucklin voting rule are in \Pshort{}. Therefore, by Proposition\nobreakspace \ref {cpmcpmw}, 
 the CPMSW and the CPMS problems are in \Pshort{} when the maximum coalition size $k=O(1)$.
\end{theorem}

\begin{proof}
 Proposition\nobreakspace \ref {cpmcpmw} says that it is enough to prove that the CPMW problem is in \Pshort{}. 
 Let $x$ be the \textit{current winner} and $y$ be the given \textit{actual winner}. 
 For any final Bucklin score $b(y)$ of $y$, there are polynomially many 
 possibilities for the positions of $x$ and $y$ in the profile of $\succ_i, i\in M$, 
 since Bucklin voting rule is anonymous. Once the positions of $x$ and $y$ is fixed, we try 
 to fill the top $b(y)$ positions of each $\succ_i^{\prime}$ - place a candidate in an empty 
 position above $b(y)$ in any $\succ_i^{\prime}$ if doing so does not make $y$ lose the election. 
 If we are able to successfully fill the top $b(y)$ positions of all $\succ_i^{\prime}$ 
 for all $i \in M$, then $M$ is a coalition of possible manipulators. 
 If the above process fails for all possible above mentioned 
 positions of $x$ and $y$ and all possible guesses of $b(y)$, then 
 $M$ is not a coalition of possible manipulators. Clearly the above algorithm runs in \textit{poly(m,n)} 
 time.
 
 The proof of correctness is as follows. If the algorithm outputs that $M$ is 
 a coalition of possible manipulators,  
 then it actually has constructed $\succ_i^{\prime}$ for all $i \in M$ with respect 
 to which they form a coalition of possible manipulators. 
 On the other hand, if they form a coalition of possible manipulators, then Lemma\nobreakspace \ref {lemmaBucklin}
 ensures that our algorithm explores all the sufficient positions of $x$ and $y$ in 
 $\succ_i^{\prime}$ for all $i \in M$. Now if $M$ is a possible coalition of manipulators, then 
 the corresponding positions for $x$ and $y$ have also been searched. Our greedy
 algorithm must find it since permuting the candidates except $x$ and $z$ which are 
 ranked above $b(y)$ in $\succ_i^{\prime}$ cannot stop $y$ to win the 
 election since the Bucklin score of other candidates except $y$ is at least $b(y)$.
\qed\end{proof}

\subsection{STV Voting Rule}

Next, we prove that the CPMW and the CPM problems for STV rule is in \NPCshort{}. To this end, we reduce from the Exact Cover by 3-Sets Problem (X3C), which is known to be in \NPCshort{} \citep{garey1979computers}. 
The X3C problem is as follows.

\begin{definition}(X3C Problem)\\
 Given a set $S$ of cardinality $n$ and $m$ subsets $S_1,S_2, \dots, S_m \subset S$ with $|S_i|=3, 
 \forall i=1, \dots, m,$ does there exist an index set $I\subseteq \{1,\dots,m\}$ 
 with $|I|=\frac{|S|}{3}$ such that $\cup_{i\in I} S_i = S$.
\end{definition}

\begin{theorem}\label{UPCMSTVNPC}
 The CPM problem for STV rule is in \NPCshort{}.
\end{theorem}

\begin{proofsketch}
 It is enough to show the theorem for the case $c=1$. 
 Clearly the problem is in \NPshort{}. To show \NPshort{} hardness, we show a many-one reduction from 
 the X3C problem to it. The reduction is analogous to the reduction given in \citep{bartholdi1991single}. 
 Hence, we give a proof sketch only. 
 Given an X3C instance, we construct an election as follows. The unspecified 
 positions can be filled in any arbitrary way. The candidate set is as follows.
 $$
 \begin{array}{rclcl}
 \mathcal{C} = \{x,y\} &\cup& \{a_1, \dots, a_m\} \cup \{\overline{a}_1, \dots, \overline{a}_m\}\\
			&\cup& \{b_1, \dots, b_m\} \cup \{\overline{b}_1, \dots, \overline{b}_m\}\\
			&\cup& \{d_0, \dots, d_n\} \cup \{g_1, \dots, g_m\}
 \end{array}
 $$
 The votes are as follows.
 \begin{itemize}
  \item $12m$ votes for $y \succ x \succ \dots$
  \item $12m-1$ votes for $x \succ y \succ \dots$
  \item $10m+\frac{2n}{3}$ votes for $d_0 \succ x \succ y \succ \dots$
  \item $12m-2$ votes for $d_i \succ x \succ y \succ \dots, \forall i\in [n]$
  \item $12m$ votes for $g_i \succ x \succ y \succ \dots, \forall i\in [m]$
  \item $6m+4i-5$ votes for $b_i \succ \overline{b}_i \succ x \succ y \succ \dots, \forall i\in [m]$
  \item $2$ votes for $b_i \succ d_j \succ x \succ y \succ \dots, \forall i\in [m], \forall j\in S_i$
  \item $6m+4i-1$ votes for $\overline{b}_i \succ b_i \succ x \succ y \succ \dots, \forall i\in [m]$
  \item $2$ votes for $\overline{b}_i \succ d_0 \succ x \succ y \succ \dots, \forall i\in [m]$
  \item $6m+4i-3$ votes for $a_i \succ g_i \succ x \succ y \succ \dots, \forall i\in [m]$
  \item $1$ vote for $a_i \succ b_i \succ g_i \succ x \succ y \succ \dots, \forall i\in [m]$
  \item $2$ votes for $a_i \succ \overline{a}_i \succ g_i \succ x \succ y \succ \dots, \forall i\in [m]$
  \item $6m+4i-3$ votes for $\overline{a}_i \succ g_i \succ x \succ y \succ \dots, \forall i\in [m]$
  \item $1$ vote for $\overline{a}_i \succ \overline{b}_i \succ g_i \succ x \succ y \succ \dots, \forall i\in [m]$
  \item $2$ votes for $\overline{a}_i \succ a_i \succ g_i \succ x \succ y \succ \dots, \forall i\in [m]$
 \end{itemize}
 The tie breaking rule is $\succ_t = \cdots \succ x$. The vote of $v$ is $x\succ \cdots$. We claim that 
 $v$ is a possible manipulator \textit{iff} the X3C is a \textit{yes} instance. 
 Notice that, of the first $3m$ candidates to be eliminated, $2m$ of them are $a_1, \dots, a_m$ and 
 $\overline{a}_1, \dots, \overline{a}_m$. Also exactly one of $b_i$ and $\overline{b}_i$ will be 
 eliminated among the first $3m$ candidates to be eliminated because if one of $b_i$, $\overline{b}_i$ 
 then the other's score exceeds $12m$.
 We show that the winner is either $x$ or $y$ irrespective of the vote of one more candidate. Let 
 $J:=\{ j: b_j \text{ is eliminated before } \overline{b}_j \}$. If $J$ is an index of set cover then the winner is 
 $y$. This can be seen as follows. Consider the situation after the first $3m$ eliminations. Let $i\in S_j$ 
 for some $j\in J$. Then $b_j$ has been eliminated and thus the score of $d_i$ is at least $12m$. 
 Since $J$ is an index of a set cover, every $d_i$'s score is at least $12m$. Notice that $\overline{b}_j$ 
 has been eliminated for all $j \notin J$. Thus the revised score of $d_0$ is at least $12m$. 
 After the first $3m$ eliminations, the remaining candidates are $x, y, \{d_i:i\in[n]\}, \{g_i:i\in[m]\}, 
 \{b_j: j\notin J\}, \{\overline{b}_j: j\in J\}$. All the remaining candidates except $x$ has score 
 at least $12m$ and $x$'s score is $12m-1$. Hence $x$ will be eliminated next which makes $y$'s score 
 at least $24m-1$. Next $d_i$'s will get eliminated which will in turn make $y$'s score $(12n+36)m-1$. 
 At this point $g_i$'s score is at most $32m$. Also all the remaining $b_i$ and $\overline{b}_i$'s score 
 is at most $32m$. Since each of the remaining candidate's scores gets transferred to $y$ once they are 
 eliminated, $y$ is the winner.
 
 Now we show that, if $J$ is not an index of set cover then the winner is $x$. This can be seen as follows. 
 If $|J|>\frac{n}{3}$, then the number of $\overline{b}_j$ that gets eliminated in the first $3m$ iterations 
 is less than $m-\frac{n}{3}$ . This makes the score 
 of $d_0$ at most $12m-2$. Hence $d_0$ gets eliminated before $x$ and all its scores gets transferred to $x$. 
 This makes the elimination of $x$ impossible before $y$ and makes $x$ the winner of the election.
 
 If $|J|\le \frac{n}{3}$ and there exists an $i\in S$ that is not covered by the corresponding set cover, then $d_i$ gets eliminated before $x$ with a score of $12m-2$ and its score gets transferred to $x$. This makes $x$ win the election.
 
 Hence $y$ can win \textit{iff} X3C is a \textit{yes} 
 instance. Also notice that if $y$ can win the election, then it can do so 
 with the voter $v$ voting a preference like $\cdots \succ x \succ y \succ \cdots$. 
\end{proofsketch}
 
From the proof of the above theorem, we have the following corollary by specifying $y$ as the \textit{actual winner} for the CPMW problem.
\begin{corollary}\label{CPMWSTVNPC}
 The CPMW problem for STV rule is in \NPCshort{}.
\end{corollary}

\section{Conclusion}\label{sec:con}

In this work, we have initiated a promising research direction for detecting manipulation in elections. We have proposed the notion of \emph{possible} manipulation and explored several concrete computational problems, which we believe to be important in the context of voting theory. These problems involve identifying if a given set of voters are possible manipulators (with or without a specified candidate winner). We have also studied the search versions of these problems, where the goal is to simply detect the presence of possible manipulation with the maximum coalition size. 
We believe there is  theoretical as well as  practical interest in studying the proposed problems. We have
provided algorithms and hardness results for many common voting rules. 

In this work, we considered elections with unweighted voters only. An immediate future research direction 
is to study the complexity of these problems in  weighted elections. Further, verifying the number 
of false manipulators that this model catches in a real or synthetic data set, where, we already have some knowledge about the manipulators, would be interesting. It is our conviction that both the problems that we have studied here have initiated an interesting research direction with significant promise and  potential for future work.

\section{Acknowledgement}

We acknowledge the anonymous reviewers of AAMAS 2014 for providing constructive comments to improve the paper.

% \small
\bibliographystyle{apalike}
\bibliography{bib_pm}

\end{document}